%% file: main.tex
\documentclass[letterpaper,10 pt, conference]{ieeeconf}

\IEEEoverridecommandlockouts                              
\usepackage{cite}

\usepackage{amsmath,amssymb,amsfonts,amsthm}
\usepackage{enumerate}
\usepackage{algorithmic}
\usepackage{graphicx}
\usepackage{textcomp}
\usepackage{xcolor}
\usepackage{pgfplots}
\usepackage{hyperref}
\def\BibTeX{{\rm B\kern-.05em{\sc i\kern-.025em b}\kern-.08em
    T\kern-.1667em\lower.7ex\hbox{E}\kern-.125emX}}

\newtheorem{theorem}{Theorem}
\newtheorem{lemma}{Lemma}
\newtheorem{corollary}{Corollary}
\newtheorem{definition}{Definition}

\DeclareMathOperator*{\argmin}{arg\,min}
\DeclareMathOperator*{\argmax}{arg\,max}
\DeclareMathOperator*{\NE}{NE}
\DeclareMathOperator*{\PoA}{PoA}
\DeclareMathOperator*{\nel}{ne}
\DeclareMathOperator*{\opt}{opt}

\begin{document}
\title{\LARGE \bf Optimal Price of Anarchy in Cost-Sharing Games\\
\thanks{This work is supported by ONR Grant \#N00014-17-1-2060, NSF Grant \#ECCS-1638214, and SNSF Grant \#P2EZP2\_181618.}}

\author{Rahul Chandan \and Dario Paccagnan \and Jason R. Marden
\thanks{R. Chandan and J.R. Marden are with the Department of Electrical and Computer Engineering, University of California, Santa Barbara, USA. Email: \href{mail_to:rchandan@ucsb.edu}{rchandan@ucsb.edu}.}
\thanks{D. Paccagnan is with the Department of Mechanical Engineering and the Center of Control, Dynamical Systems and Computation, UC Santa Barbara, USA. Email: \href{mail_to:dariop@ucsb.edu}{dariop@ucsb.edu}.}}

\maketitle

\begin{abstract}
    The design of distributed algorithms is central to the study of multiagent systems control. 
    In this paper, we consider a class of combinatorial cost-minimization problems and propose a framework for designing distributed algorithms with a priori performance guarantees that are near-optimal. 
    We approach this problem from a game-theoretic perspective, assigning agents cost functions such that the equilibrium efficiency (price of anarchy) is optimized. 
    Once agents' cost functions have been specified, any algorithm capable of computing a Nash equilibrium of the system inherits a performance guarantee matching the price of anarchy. 
    Towards this goal, we formulate the problem of computing the price of anarchy as a tractable linear program. 
    We then present a framework for designing agents' local cost functions in order to optimize for the worst-case equilibrium efficiency. 
    Finally, we investigate the implications of our findings when this framework is applied to systems with convex, nondecreasing costs.
\end{abstract}

\input{introduction.tex}

\input{relatedworks.tex}
\input{illustration.tex}
\input{characterizing.tex}
\input{optimizing.tex}
\input{specialcase.tex}

\section{Conclusions}
In this paper, we investigated distributed algorithm design for finding approximate solutions to a combinatorial problem that is difficult to solve, in general. 
By approaching the problem from a game-theoretic perspective, we were able to develop a linear program approach to calculate the price of anarchy corresponding to a specified system objective. 
Next, we proposed a method of designing the distribution rule to minimize the price of anarchy for a given system. 
For the case of convex and nondecreasing cost functions, we have shown how the linear program further simplifies. 
Additionally, we have demonstrated that our designed distribution rule has superior performance over the two most commonly used (i.e. Shapley value and marginal contribution).

\bibliographystyle{IEEEtran}
\bibliography{references}
\null \vfill
\appendix 
\input{appendix.tex}

\end{document}

%% file: introduction.tex
\section{Introduction}
The study of multiagent systems control has become more popular in recent years, with the advent of a variety of technologies that promise to change the way we interact with our surroundings and significantly improve our quality of life. In the near future, we envision fleets of autonomous driverless cars and unmanned aerial vehicles that will allow the transportation and shipping industries to reduce their cost and environmental impact dramatically \cite{spieser2014toward}. As an alternative example, advances in the distributed control of the power grid and transportation networks could potentially make our lives more affordable and comfortable, and extend the expected lifespan of existing infrastructure \cite{moradipari2018electric,leontiadis2011effectiveness,ma2013decentralized}. These and many other industries will benefit from advancements in our ability to control multiagent systems \cite{alexis2009coordination,kuntze2012seneka,servant2015controlled}.

One of the main challenges of controlling multiagent systems is in the design of control algorithms that can generate decision-making rules for the agents of the system. We seek algorithms that give rise to the desired global behaviour, and that do so reliably and within the system's constraints. In order to satisfy these requirements, it is often impractical for a central controller to communicate with and steer every individual agent in the network. Instead, we use distributed protocols, which would more easily satisfy a system's scalability, communication bandwidth, and privacy requirements.

Game theory provides an alternative approach to the problem of distributed system design \cite{bauso2016game}. Although game theory originated as a discipline in economics, it has quickly taken a foothold in the study of multiagent systems, as methodologies for the design of distributed protocols are readily available that have provable performance guarantees on the emergent system behaviour. Similar approaches have been applied to a wide range of design problems including communication network design \cite{song2011optimal,ahmad2009spectrum}, demand/response balancing \cite{ma2013decentralized}, optimal power generation \cite{gebraad2016wind}, and social influence \cite{brown2017studies}.

In the game-theoretic approach, a given optimization problem is solved in a distributed fashion by posing the problem as a game. This approach, as proposed in \cite{arslan2007autonomous,shamma2008cooperative}, can be summarized in two steps. First, the system designer assigns local cost functions to the agents in the system, then the agents' decision-making rules are formulated such that the overall system is driven to the desired equilibria. If, after these two steps, the equilibrium allocations optimize (or are provably close to) the system-level objective, the end result of this game-theoretic procedure is an efficient and distributed solver for the original optimization problem.

We seek to establish a framework that enables system designers to develop distributed algorithms that promote provably near-optimal system behaviour that is agnostic to the system's structure. Our work leverages recent results in \cite{paccagnan2018distributed} where the authors make significant headway towards achieving this objective by providing methodologies for tightly (exactly) calculating the worst-case performance of equilibria for a specified class of games. This is done by posing the problem of calculating the so-called price of anarchy as a linear program, and then providing a tractable solution to the problem of finding the decision-making rules that maximize the performance of the game, using the price of anarchy as a metric. While the results in \cite{paccagnan2018distributed} are limited to the class of welfare maximization problems, in this paper, we consider cost-minimization objectives.

\input{problemformulation.tex}

\begin{figure*}[t]
\centering
\input{plot.tex}
\quad
\begin{tabular}[b]{ccc} \hline
    $d$ & $\frac{\PoA(f_{\textrm{SV}})}{\PoA(f^*)}$  & $\frac{\PoA(f_{\textrm{MC}})}{\PoA(f^*)}$\\ \hline 
    1   & 1     & 1 \\ 
    1.2 & 1.03  & 1.151 \\ 
    1.4 & 1.069 & 1.277 \\ 
    1.5 & 1.092 & 1.33 \\ 
    1.6 & 1.117 & 1.376 \\ 
    1.8 & 1.174 & 1.447\\ 
    2   & 1.242 & 1.491\\ \hline
    \multicolumn{3}{c}{\vspace{16pt}}\\
\end{tabular}
\quad     
\begin{tabular}[b]{cccc} \hline
    $j$ & $f^*$         & $f^{\textrm{MC}}$     & $f^{\textrm{SV}}$ \\ \hline 
    1   & 1             & 1                     & 1                 \\ 
    2   & 0.484         & 0.564                 & 0.5               \\ 
    3   & 0.318         & 0.385                 & 0.333             \\ 
    4   & 0.236         & 0.291                 & 0.25              \\ 
    5   & 0.189         & 0.234                 & 0.2               \\ 
    6   & 0.157         & 0.196                 & 0.167             \\ 
    7   & 0.134         & 0.168                 & 0.143             \\ 
    \vdots & \vdots & \vdots & \vdots\\    \hline
    \multicolumn{3}{c}{\vspace{1pt}}\\
\end{tabular}
\vspace{-10pt}
\caption{Price of anarchy of the optimal distribution rule $f^*$ plotted against $\PoA(f_{\textrm{MC}})$ and $\PoA(f_{\textrm{SV}})$. This game comprises $n = 20$ agents and a cost function of the form $c(j) = j^d \text{, }d\in [1,2]$. The exact prices of anarchy of the distribution rules were calculated using the tractable linear program developed in Section \ref{sec:lin-prog}. Observe that the optimal distribution rule found using the linear program in Theorem \ref{thm:opt_dist_rule} clearly outperforms the two other distribution rules, especially for cost functions of higher order (i.e. greater $d$). The left table delineates the relative efficiency of the other distribution rules when compared to the one found using our framework for varying values of $d$. The table on the right shows computed values of the three distribution rules when $d = 1.2$.}
\vspace{6pt}
\hrulefill
\vspace{-18pt}
\label{fig:poa-chart}
\end{figure*}

\subsection{Our Contributions}
The objectives of this paper are two-fold. First, we intend to develop a tractable method for tightly computing the price of anarchy within the class of problems considered. Second, we wish to design agents' cost functions of the form \eqref{eq:player-cost} such that the price of anarchy is minimized, i.e. we wish to find $f^*$ such that,
\begin{equation}\label{eq:objective}
    f^* = \argmin_{f \in \mathbb{R}^n} \PoA(f).
\end{equation}

This paper is inspired by the results in \cite{paccagnan2018distributed}, in which a linear program solution to the optimization of price of anarchy in welfare maximization games is presented. Though we focus on cost-minimization games in this work, many of our results extend from that paper. 

The main contributions of this paper are the following:
\begin{itemize}
    \item We demonstrate that the problem of calculating the exact price of anarchy of a given cost-minimization game (i.e. solving \eqref{eq:poa}) can be reformulated as a tractable linear program (Theorems \ref{thm:poa_as_lp} and \ref{thm:dual}) with two unknowns and $\mathcal{O}(n^2)$ constraints.
    \item We show that the price of anarchy can be expressed explicitly as the minimum over $\mathcal{O}(n^2)$ computations for special classes of the distribution rule (Corollary \ref{cor:nondecr}, Theorem \ref{thm:c_convex}).
    \item We prove that the problem of finding the distribution rule that minimizes the price of anarchy (i.e. \eqref{eq:objective}) is also a tractable linear program (Theorem \ref{thm:opt_dist_rule}) with $n+1$ unknowns and $\mathcal{O}(n^2)$ constraints.
\end{itemize}

%% file: problemformulation.tex
\subsection{Problem Formulation} \label{sec:prob-form}
In this paper, we consider cost-minimization games with agents $N = [n] = \{1,\dots,n\}$, and $m$ resources in the set $\mathcal{R} = \{r_1, \dots, r_m\}$. Every resource $r$ is associated with an anonymous local cost function $c_r(\cdot):N\to \mathbb{R}$. We consider anonymous local cost functions of the form $c_r(j) = v_r \cdot c(j)$ for any $1{\leq}\! j {\leq}\! n$, where $v_r {\geq}\! 0$ is the relative value associated with resource $r$, and $c(\cdot):N\to \mathbb{R}_{>0}$ is the system's base cost function. We normalize this cost function such that $c(1) = 1$.

Each agent $i \in N$ selects an allocation $a_i$ from the corresponding action set $\mathcal{A}_i \subseteq 2^{\mathcal{R}}$, and the tuple $a = (a_1, \dots, a_n)$ contains the selections of all the agents. The system cost of allocation $a \in \mathcal{A} = \mathcal{A}_1 \times \dots \times \mathcal{A}_n$ is,
\begin{equation}\label{eq:system-cost}
    C(a) = \sum_{r \in \mathcal{R}} c_r(|a|_r) = \sum_{r \in \mathcal{R}} v_rc(|a|_r)\text{,}
\end{equation} 
where $|a|_r = |\{i \in N \text{ s.t. } r \in a_i\}|$ is the number of agents selecting resource $r$ in allocation $a$. We are interested in the allocation, $a^{\opt}$, that minimizes the system cost, i.e.,
\begin{equation} \label{eq:optimalallocation}
 a^{\opt} = \argmin_{a \in \mathcal{A}} C(a). 
\end{equation}

Ideally, we would design a distributed solution to \eqref{eq:optimalallocation}. Unfortunately, finding $a^\textrm{opt}$ is an intractable (NP-hard) problem, in general. Therefore, we settle for developing distributed and tractable algorithms that find approximate solutions, and with guarantees on the approximation ratio, using the game design approach. The system designer is free to assign agents' cost functions, $J_i(\cdot):\mathcal{A}\to \mathbb{R}$. In order to preserve the distributed nature of the system, an agent cost function $J_i$ should only depend on the information that is locally available to agent $i$. We use the following model for the cost of allocation $a$ to a given agent $i \in N$,
\begin{equation} \label{eq:player-cost}
 J_i(a) = \sum_{r \in a_i} v_rc(|a|_r)\cdot f(|a|_r) \text{,}
\end{equation}
where we refer to $f(\cdot):N\to \mathbb{R}_{\geq 0}$ as the distribution rule. The distribution rule represents the proportion of $v_rc(|a|_r)$ that every agent selecting $r$ experiences, and it is the only system parameter that the system designer can modify\footnote{More general forms of the agent cost function exist that are nonanonymous, but recent results (e.g. \cite[Thm. 2]{gairing2009covering}) suggest that these perform just as poorly in the worst-case as anonymous cost functions.}. We underline that, although we formalize the player cost function as \eqref{eq:player-cost}, there are no constraints on the distribution rule $f(\cdot)$.

We represent the above game with the tuple $G = (\mathcal{R}, \{v_r\}_{r \in \mathcal{R}}, N, \{\mathcal{A}_i\}_{i \in N}, f)$. In order to simplify the notation, we omit the subscripts for the value and action sets. 

We can now define the notion of Nash equilibrium in cost minimization games,
\begin{definition}[Nash equilibrium]
A given allocation $a^{\nel} \in \mathcal{A}$ is a pure Nash equilibrium of the cost-minimization game $G$ if $J_i(a^{\nel}) \leq J_i(a_i, a_{-i}^{\nel})$ for all allocations $a_i \in \mathcal{A}_i$ and all agents $i \in N$. $a_{-i}$ denotes all the entries in the tuple $a$ except $a_i$, i.e. $a_{-i} = (a_1, \dots, a_{i-1}, a_{i+1}, \dots, a_n)$.
\end{definition}

The goal of this paper is to find the distribution rule $f$ with which the multiagent system performs near-optimally for a class of games 
$\mathcal{G}_f = \{(\mathcal{R}, \{v_r\}, N, \{\mathcal{A}_i\}, f) \}$. 
Note that specifying a distribution rule $f(\cdot)$ results in a well-defined game for any resource set $\mathcal{R}$, resource valuation $\{v_r\}$, and action set $\{\mathcal{A}_i\}$, hence providing the desired scalability. The performance of a distribution rule $f$ will be measured using the so-called price of anarchy, which is defined as,
\begin{equation}\label{eq:poa}
    \PoA(f) = \sup_{G \in \mathcal{G}_f} \frac{\max_{a \in \NE(G)} C(a)}{\min_{a \in \mathcal{A}} C(a)} \geq 1,
\end{equation} 
where $\NE(G)$ is the set of all Nash equilibria of the game $G$\footnote{To capture all of its dependencies, price of anarchy should be denoted as $\PoA(f; N,c)$, here we use $\PoA(f)$ to simplify the notation.}. While the system cost $C(\cdot)$ in the above also depends on the game $G$, we do not explicitly indicate this for ease of presentation. A price of anarchy close to $1$ is desirable as the worst-case equilibrium efficiency nears optimality.

The price of anarchy is an approximation ratio for the worst-case performance of any algorithm capable of calculating Nash equilibria. In general, the computation of Nash equilibria is a PLS complete problem, but they can be computed in polynomial time for certain classes of games, see \cite{ackermann2008impact} for an illustration. 

%% file: plot.tex
\begin{tikzpicture}

\begin{axis}[
width=8cm, height=5cm,
axis background/.style={fill=white},
axis line style={black},
tick align=outside,
tick pos=left,
legend style={font=\small},
legend pos=north west,
x grid style={gray},
    xlabel      = $d$, 
xmajorgrids,
xmin=1, xmax=2,
y grid style={gray},
ymajorgrids,
ymin=1, ymax=3
]
\addplot [line width=0.75pt, color=black, mark=x, mark size=3]
table [row sep=\\]{%
1	1 \\
1.05	1.029548028 \\
1.1	1.060490371 \\
1.15	1.092991737 \\
1.2	1.127281336 \\
1.25	1.163385918 \\
1.3	1.201403239 \\
1.35	1.241434105 \\
1.4	1.283631136 \\
1.45	1.32807416 \\
1.5	1.374948439 \\
1.55	1.424339107 \\
1.6	1.476450613 \\
1.65	1.53141702 \\
1.7	1.589395553 \\
1.75	1.650600819 \\
1.8	1.715207025 \\
1.85	1.783453122 \\
1.9	1.855494118 \\
1.95	1.93162063 \\
2	2.012072435 \\
};
\addlegendentry{$\PoA(f^*)$}

\addplot [line width=0.75pt, color=green!50!black, mark=triangle, mark size=2]
table [row sep=\\]{%
1	1 \\
1.05	1.03597957 \\
1.1	1.074667928 \\
1.15	1.116196004 \\
1.2	1.160712213 \\
1.25	1.2083862\\
1.3	1.25939826\\
1.35	1.313939585\\
1.4	1.372288016\\
1.45	1.43465848\\
1.5	1.501366243\\
1.55	1.572722698\\
1.6	1.649103712\\
1.65	1.730912364\\
1.7	1.818545527\\
1.75	1.912557855 \\
1.8	2.013490386 \\
1.85	2.12192586 \\
1.9	2.238638908\\
1.95	2.364345667\\
2	2.5\\
};
\addlegendentry{$\PoA(f_{\textrm{SV}})$}

\addplot [line width=0.75pt, color=blue, mark=*, mark size=2]
table [row sep=\\]{%
1	1 \\
1.05	1.070526271 \\
1.1	1.143549807 \\
1.15	1.219140506 \\
1.2	1.297403895 \\
1.25	1.378416751 \\
1.3	1.462287603 \\
1.35	1.549114681 \\
1.4	1.639021832 \\
1.45	1.732081616 \\
1.5	1.828420976 \\
1.55	1.928156875 \\
1.6	2.031446796 \\
1.65	2.138351331 \\
1.7	2.249010435 \\
1.75	2.363563308 \\
1.8	2.482190285 \\
1.85	2.604980723 \\
1.9	2.732165788 \\
1.95	2.86377044 \\
2	3.00003 \\
};
\addlegendentry{$\PoA(f_{\textrm{MC}})$}

\end{axis}

\end{tikzpicture}

%% file: relatedworks.tex
\subsection{Related Works}\label{sec:lit-rev}
Over the past two decades, the characterization of the price of anarchy has been a research focus in the field of algorithmic game theory \cite{koutsoupias1999worst,nisan2007algorithmic}. Several recent advances have been made towards computing the price of anarchy of games and finding the optimal distribution rule. A popular method for lower-bounding the price of anarchy comes from smoothness arguments, e.g., see \cite{roughgarden2015intrinsic}. Unfortunately, this method is limited in its applicability to game design problems, as lower-bounds coming from smoothness are only tight for congestion games in which the sum over all agents' local costs is equal to the system cost (the budget-balanced constraint). For a formal proof smoothness arguments are loose for design problems, see \cite[Thm. 1]{paccagnan2018distributed}.

The work in \cite{marden2014generalized} provides a lower bound to the performance of a more general class of equilibria, but is limited to resource-allocation games using the Shapley value distribution rule (defined in Section \ref{sec:illustrat}) and in which agents can only have singleton strategies.

A novel distribution rule is presented in \cite{gairing2009covering}, and a lower-bound on its price of anarchy is found. Note that their results apply to more general notions of cost function and equilibrium than those considered in this paper. The authors of \cite{ramaswamy2017impact} complete the proof of tightness for the new distribution rule by constructing a game with price of anarchy equal to this lower-bound. In a general setting, it is inefficient to approach the design of distribution rules by examining their suitability for specific applications one at a time.

The goal of designing agents' cost functions that optimize for effiency metrics including the price of anarchy has already received much attention. Past application domains have included resource allocation \cite{gkatzelis2016optimal,von2013optimal,vetta2002nash,jensen2018optimal}, set covering \cite{gairing2009covering,philips2016importance}, communication networks \cite{chen2010designing}, and congestion pricing problems \cite{sandholm2002evolutionary,brown2018optimal}. For many of these problems, the aforementioned Shapley value distribution rule has been found to be optimal within the assumptions and constraints, even for nonanonymous setups. However, we are currently unable to design cost functions that optimize over efficiency measures in a general context.

Many works have explored the adverse effects of optimizing over the price of anarchy on the efficiency of the equilibria of a game by examining the tradeoffs with the so-called \textit{price of stability} (e.g. \cite{ramaswamy2017impact,von2013optimal,phillips2018design}). Their results suggest that focusing solely on optimizing the price of anarchy of a game (i.e. worst-case equilibrium efficiency) will reduce the efficiency of any equilibria with superior performance. The analysis of this tradeoff is important for our framework but is outside the scope of this paper.

%% file: illustration.tex
\subsection{An illustration of our main result} \label{sec:illustrat}
We begin with a simple example to illustrate the results in this paper. Consider the problem in which we seek to coordinate a multiagent system with $n = 20$ agents, and convex and nondecreasing cost functions $c(j) = j^d\text{, } d\in [1,2]$. For this illustration, we borrow the two most widely studied distribution rules in the literature, the Shapley value \cite{von2013optimal,jensen2018optimal,chen2010designing,harks2014optimal} and marginal contribution \cite{ramaswamy2017impact,philips2016importance}, in order to benchmark the relative performance of the optimal distribution rule found using our proposed framework.
\begin{definition}[Distribution rules]
The Shapley value, $f_{\textrm{SV}}$, and marginal contribution $f_{\textrm{MC}}$ are defined for $j \in N$ as,
\begin{align}
    f_{\textrm{SV}}(j) &= \frac{1}{j}\text{, }\\
    f_{\textrm{MC}}(j) &= 1 - \frac{c(j-1)}{c(j)}\text{.}
\end{align}
\end{definition}
The Shapley value distribution rule is budget-balanced, which means that the player and system cost functions defined in \eqref{eq:system-cost} and \eqref{eq:player-cost} satisfy $C(a) = \sum_{i \in N} J_i(a)$. On the other hand, the marginal contribution distribution rule ensures $J_i(a) = C(a) - C(\emptyset,a_{-i})$, where $\emptyset$ denotes the null allocation. While our study is not limited to budget-balanced distribution rules, there are a number of recent results suggesting that the Shapley value has optimal price of anarchy over such rules \cite{von2013optimal,jensen2018optimal,chen2010designing,harks2014optimal}.

In Figure \ref{fig:poa-chart}, we compare the price of anarchy for the multiagent system when agents use the optimal distribution rule (Theorem \ref{thm:opt_dist_rule}) with the price of anarchy for the Shapley value and marginal contribution. Using Theorem \ref{thm:c_convex}, closed-form expressions for tightly characterizing the prices of anarchy for these two distribution rules were found. As seen in Figure \ref{fig:poa-chart}, the optimal distribution rule calculated using our linear program has equal or better price of anarchy when compared to $\PoA(f_{\textrm{SV}})$ and $\PoA(f_{\textrm{MC}})$, with significant relative improvement as $d$ increases (i.e. higher order cost functions). For example, when $d = 1.8$, the price of anarchy corresponding to the optimal distribution rule performs $1.17\times$ better than the Shapley value, and has $1.45\times$ improvement on the marginal contribution.

%% file: characterizing.tex
\section{Characterizing the price of anarchy} \label{sec:lin-prog}

In this section, we reformulate the problem of finding the price of anarchy for a given family of games $\mathcal{G}_f$ as a linear program. The program's solution is proven to be tight. This result will allow us to also pose the distribution function design problem as a linear program.

The following assumptions are necessary in order to ensure that we always have $C(a^{\textrm{opt}}) > 0$, which ensures that \eqref{eq:poa} is well-defined. 

\textbf{Standing Assumptions}
There is at least one agent $i \in N$ that has action set $\mathcal{A}_i$ such that $v_r > 0 \text{, } \forall r \in a_i \in \mathcal{A}_i$. With slight abuse of notation, we extend the definitions of $f$ and $c$ for $j=0$ and $j=n+1$, and assign that $f(0) = c(0) = 0$, $f(n+1) = f(n)$, and $c(n+1) = \infty$. We note that this does not change the results, but simplifies the notation.

A parametrization of the various allocations in the game, borrowed from \cite{ward2012oblivious}, is central to the development of our linear program. In order to parametrize the problem, we define the following sets.

\begin{definition}
We define the sets $\mathcal{I}$ and $\mathcal{I_R}$ as follows,
\begin{align*}
    \mathcal{I} &:= \{(a,x,b) \in \mathbb{N}_{\geq 0}^{3} \text{ s.t. } 1 \leq a + x + b \leq n \}, \\
    \mathcal{I_R} &:= \{(a,x,b) \in \mathcal{I} \text{ s.t. } a\cdot x\cdot b=0 \text{ or } a + x +b = n \}.
\end{align*} 
\end{definition}

Note that $\mathcal{I_R}$ contains all points $(a,x,b)$ on the planes bounding its superset $\mathcal{I}$, i.e., $a\!{=}0$, $x\!{=}0$, $b\!{=}0$, and $a\!{+}x\!{+}b\!{=}n$.


\begin{theorem}\label{thm:poa_as_lp}
    Given the family of games $\mathcal{G}_f\!{=} \{(\mathcal{R}, \{v_r\}, N, \{\mathcal{A}_i\}, f)\}$, the price of anarchy is,
    \[ \PoA(f) = 1/C^* \text{,}\]
    where $C^*$ is the solution to the following linear program,
    \begin{align}
        &C^* = \min_{\theta(a,x,b)} \sum_{a,x,b} 1_{\{b+x \geq 1\}}c(b+x)\theta(a,x,b) \\
        \text{s.t.} &{\sum_{a,x,b}}\! \Big[af(a{+}\!x)c(a{+}\!x) {-}\! bf(a{+}\!x{+}\!1)c(a{+}\!x{+}\!1)\Big]\theta(a,x,b) {\leq} 0 \nonumber\\
        &\sum_{a,x,b} 1_{\{a+x \geq 1\}}c(a{+}\!x)\theta(a,x,b) {=}\: 1 \nonumber\\
        &\theta(a,x,b) \geq 0 \quad \forall (a,x,b) \in \mathcal{I}\text{.} \nonumber
    \end{align}
\end{theorem}
\begin{proof} 
    The first observation is that the game $G$ with worst-case Nash equilibrium $a^{\nel}$ and optimal allocation $a^{\opt}$ has the same price of anarchy as the game $\hat{G}$ in which every agent $i \in N$ only has action set $\mathcal{A}_i = \{a_i^{\nel}, a_i^{\opt}\}$. We can therefore reduce the set of games $\mathcal{G}_f$ to the set $\hat{\mathcal{G}}_f$ and write,
    \begin{align*}
        \PoA(f) = {\sup_{G \in \hat{\mathcal{G}_f}}}\!& \frac{C(a^{\nel})}{C(a^{\opt})}\\
        \text{ s.t. }& J_i(a^{\nel}) {\leq}\: J_i(a_i^{\opt}, a_{-i}^{\nel}) \text{ } \forall i {\in}\: N\text{.}
    \end{align*}
    By assumption, there is at least one agent in the game covering a resource with nonzero value, which must mean that $C(a^{\nel}) > 0$. Therefore, without loss of generality, we can normalize the values such that $C(a^{\nel}) = 1$. We also relax the constraint on the definition of Nash equilibrium,
    \begin{align*}
        D^* = \inf_{G \in \hat{\mathcal{G}_f}} &C(a^{\opt})\\
        \text{ s.t. } & \sum_{i \in N} J_i(a^{\nel}) \leq \sum_{i \in N} J_i(a_i^{\opt}, a_{-i}^{\nel})\text{, }  C(a^{\nel}) = 1\text{,}
    \end{align*}
    where $\PoA(f) = 1/D^*$. We know that $D^* \leq C^*$ because of the relaxation, but we show in Lemma \ref{lem:const_geq} (found in the Appendix) that $D^* \geq C^*$ which implies equality. We now show that any game can be posed under the parametrization $\theta(a,x,b) \in \mathbb{R}, (a,x,b) \in \mathcal{I}$. For every triplet $(a,x,b) \in \mathcal{I}$, we assign to $\theta(a,x,b)$ the sum over the values of all resources $r \in \mathcal{R}$ for which exactly $x$ agents have $r \in a^{\textrm{ne}}$ and $r \in a^{\textrm{opt}}$, exactly $a$ agents have $r \in a^{\textrm{ne}}$ and $r \notin a^{\textrm{opt}}$, and exactly $b$ agents have $r \notin a^{\textrm{ne}}$ and $r \in a^{\textrm{opt}}$. The parametrized version of the problem appears as follows,
    \begin{align*}
        &D^* = \inf_{\theta(a,x,b)} \sum_{a,x,b} 1_{\{b+x\geq1\}} c(b\!{+}x)\theta(a,x,b) \\
        \text{s.t.} & {\sum_{a,x,b}}\! [af(a\!{+}x)c(a\!{+}x) {-}\! bf(a\!{+}x\!{+}1)c(a\!{+}x\!{+}1)]\theta(a,x,b) {\leq}\: 0, \\ 
        & \sum_{a,x,b} 1_{\{a+x\geq1\}} c(a\!{+}x)\theta(a,x,b) {=}\: 1, \\
        & \theta(a,x,b) {\geq}\: 0 \quad \forall (a,x,b) \in \mathcal{I}\text{.}
    \end{align*}
    The infimum must be attained because when $a+x$ is greater than zero, $\sum_{a,x,b} 1_{\{a+x\geq1\}} c(a+x)\theta(a,x,b) = 1$, and when $a+x$ is zero, $\theta(0,0,b)$ must be bounded because, $\sum_{b \in N} bf(1)c(1)\theta(0,0,b) \geq 0$ and, because $C(a^{\opt}) \leq C(a^{\nel}) = 1$. 
    This is in fact a linear program of the form, 
    \begin{align*}
        C^* =  \min_y &c^\top y \\
        \text{ s.t. } &e^\top y \leq 0 \text{, } d^\top y - 1 = 0 \text{, } -y \leq 0\text{.}
    \end{align*} 
    where $y$ is a column vector of all $(n+1)(n+2)(n+3)/6$ unkowns represented by $\theta(a,x,b) \text{, } \forall (a,x,b) \in \mathcal{I}$.
\end{proof}

We note that, given a cost function $c$ and a distribution rule $f$, the above linear program returns the price of anarchy as well as the worst-case equilibrium and optimal allocation, encoded in $\theta(a,x,b)$. 
By rewriting the linear program in its Lagrangian dual, and by exploiting the problem structure to reduce the number of constraints, we reformulate the primal as a linear program that has $\mathcal{O}(n^2)$ constraints and only two unknowns.
\begin{theorem}[The dual reformulation]\label{thm:dual}
    For the family of games $\mathcal{G}_f$, the price of anarchy is $1/C^*$, where $C^*$ is the solution of the following program,
    \begin{align*}
        &C^* = \max_{\lambda \in \mathbb{R}_{\geq0}, \mu \in \mathbb{R}} \mu \\
        \text{s.t. } &1_{\{b+x\geq1\}}c(b+x) - \mu 1_{\{a+x\geq1\}} c(a+x) \\
        & + \lambda \Big[af(a\!{+}x)c(a\!{+}x)- bf(a\!{+}x\!{+}1)c(a\!{+}x\!{+}1)\Big] \geq 0 \\
        & \forall (a,x,b) \in \mathcal{I_R}\text{.}
    \end{align*}
\end{theorem}
\begin{proof}
We set up the Lagrangian function $L(\lambda, \mu, \nu, a,x,b)$, with Lagrange multipliers $\lambda \in \mathbb{R}_{\geq 0}$, $\mathbf{\nu} \in \mathbb{R}^\mathcal{I}_{\geq 0}$, $\mu \in \mathbb{R}$, and define the dual function $g(\lambda, \mu, \mathbf{\nu})$ as,
\begin{align*}
    g&(\lambda, \mu, \nu) = \inf_{a,x,b} L(\lambda, \mu, \nu,a,x,b) \\
    &= \mu + \inf_{a,x,b} \sum_{a,x,b}\Big(1_{\{b+x\geq1\}}c(b+x) + \lambda e(a,x,b) \\
    & \qquad - \mu 1_{\{a+x\geq1\}} c(a+x) - \nu_{a,x,b}\Big)\theta(a,x,b)\text{,}\\
    &= \begin{cases}
    \mu &\text{if } 1_{\{b+x\geq1\}}c(b+x)+ \lambda e(a,x,b)\\
    & \quad - \mu 1_{\{a+x\geq1\}} c(a+x)  = \nu_{a,x,b} \geq 0 \\ 
    -\infty &\text{otherwise.}
    \end{cases}
\end{align*}
where $e(a,x,b) = af(a\!{+}x)c(a\!{+}x)- bf(a\!{+}x\!{+}1)c(a\!{+}x\!{+}1)$ for all $(a,x,b) \in \mathcal{I}$. The dual problem is,
\begin{align*}
    \max_{\lambda\geq0, \mu}& \mu \text{ s.t. } 1_{\{b+x\geq1\}}c(b+x) - \mu 1_{\{a+x\geq1\}} c(a+x) \\
    & \quad + \lambda e(a,x,b) \geq 0, \forall (a,x,b) \in \mathcal{I}
\end{align*}
We know that we have strong duality because the primal problem is linear. After defining the changes of variables $j = a+x$ and $l=b+x$, the constraint can be written as, 
\begin{align*}
    \mu c(j) &\leq c(l) + \lambda[(j-x)f(j)c(j)-f(j+1)c(j+1)(l-x)] \\
    &= c(l) + \lambda\Big[jf(j)c(j)-lf(j+1)c(j+1) \\
    &\qquad \qquad \qquad+ x[f(j+1)c(j+1)-f(j)c(j)]\Big].
\end{align*}
When $j$ and $l$ are held constant, if $f(j)c(j)$ is decreasing, we want $x$ to be as big as possible to tighten the constraint on $\mu$. Due to the same reasoning, if $f(j)c(j)$ is nondecreasing, we want $x$ to be as small as possible for constant $j$ and $l$. The value of $x$ is constrained by $x\geq 0$, $x\geq j+l-n$, $x\leq l$,$x \leq j$, and $x \leq j+l-1$. Thus,
\[
x = \begin{cases}
    \min\{j,l\} & \text{if } f(j+1)c(j+1) < f(j)c(j) \\
    \max\{0, j+l-n\} & \text{if } f(j+1)c(j+1) \geq f(j)c(j)\text{.}
\end{cases}
\]
We now show that for all possible $j$ and $l$, $(a,x,b)$ resides in the set $I_R$.
\begin{enumerate}
    \item When $f(j)c(j)$ is decreasing, $x=j$ when $j\leq l$ which implies that $a=0$, and $x=l$ when $j>l$ implying that $b=0$. In both cases, $a\cdot x\cdot b = 0$, and so $(a,x,b)\in \mathcal{I_R}$.
    \item When $f(j)c(j)$ is nondecreasing, $x=0$ when $j+l \leq n$, and when $j+l > n$, setting $x=j+l-n$ implies that $a+b+x=n$. These cases are included in $a\cdot x\cdot b=0$, and $a+x+b=n$, respectively. Therefore, $(a,x,b)\in \mathcal{I_R}$.
\end{enumerate}
\end{proof}

According to the previous theorem's proof, there is potential for simplification if we examine $f(j)w(j)$ that are nondecreasing and decreasing, separately. The following corollary demonstrates that the (dual) linear program can be rewritten with exactly $n^2$ constraints, and no more than $2$ unknowns for nondecreasing $f(j)w(j)$. A similar result exists for decreasing $f(j)w(j)$, but is left out for brevity.

\begin{corollary}\label{cor:nondecr}
    Consider the family of games $\mathcal{G}_f$.
    \begin{enumerate}[i)]
        \item If $f(j)c(j)$ is nondecreasing for all $j \in N$, the dual program further simplifies such that, for $\PoA = 1/C^*$,
        \begin{align*}
            &C^*= \max_{\lambda \in \mathbb{R}_{\geq 0}, \mu \in \mathbb{R}} \mu \\
            \text{s.t. } &\mu c(j) \leq c(l) \!{+} \lambda\Big[jf(j)c(j) \!{-} lf(j\!{+}1)c(j\!{+}1)\Big] \\
            & \qquad \forall j,l\in [0,n], 1 \leq j+l \leq n,\\
            &\mu c(j) \leq c(l) \!{+} \lambda\Big[(n\!{-}l)f(j)c(j) \!{-}(n\!{-}j)f(j\!{+}1)c(j\!{+}1)\Big], \\
            & \qquad \forall j,l\in [0,n], j+l > n.
        \end{align*}
        \item Additionally, if $f(j) \leq \frac{1}{j}f(1)c(1)\max_{l\in N}\frac{l}{c(l)}$, $\forall j \in N$,
        \[ \lambda^* = \frac{1}{f(1)c(1)}\min_{l\in N} \frac{c(l)}{l}. \]
    \end{enumerate}
\end{corollary}
\begin{proof}
    \textbf{i)} From the proof of Theorem \ref{thm:dual}, we know that $x=0$ when $1\leq j+l \leq n$, $j,l \in [0,n]$. This gives,
    \[ \mu c(j) \leq c(l) + \lambda[jf(j)c(j) - lf(j+1)c(j+1)] \text{.}\]
    We know from the same proof that when $j+l > n$, $j,l \in [0,n]$, then $x=j+l-n=n-a-b$, i.e.,
    \[ \mu c(j) \leq c(l) \!{+} \lambda[(n\!{-}l)f(j)c(j) \!{-} (n\!{-}j)f(j\!{+}1)c(j\!{+}1)]\text{.} \]
    \textbf{ii)} We begin by observing that the constraints on the dual program corresponding to \(j=0\) are,
  \[ \lambda \leq \frac{c(l)}{l}\frac{1}{f(1)c(1)}\text{, }\forall l \in N,\]
  since \(j+l = l \leq n\), and part i) of this corollary. We denote the strictest of these bounds on \(\lambda\) as,
  \[\lambda^* = \frac{1}{f(1)c(1)}\min_{l \in N}\frac{c(l)}{l}.\]
  When \(l=0\), the constraint in part i) simplifies to \(\mu \leq \lambda jf(j)\). The rest of this proof amounts to showing that the maximum value of \(\mu\) is at \(\lambda^*\). We consider the cases \(1 \leq j+l \leq n\) and \(j+l > n\) separately.
  
  Case \(1 \leq j+l \leq n\): Here we show that the constraints corresponding to \(j > 0\) and \(l = 0\) are stricter on \(\mu\) for \(\lambda \leq \lambda^*\) than any other constraints, i.e.,
  \[ \lambda jf(j) \leq \frac{c(l)}{c(j)} + \frac{\lambda}{c(j)}[jf(j)c(j)-lf(j+1)c(j+1)], \]
  whis is equivalent to showing that,
  \[ \frac{c(l)}{c(j)} - \frac{\lambda}{c(j)}[lf(j+1)c(j+1)] \geq 0. \]
  When \(j \geq l > 0\) and since \(f(j)c(j)\) is nondecreasing, \(f(j+1)c(j+1) \geq f(j)c(j) \geq f(l)c(l)\), which supports the first inequality in,
  \begin{IEEEeqnarray}{Rl}
    &\frac{c(l)}{c(j)} - \frac{\lambda}{c(j)}[lf(j+1)c(j+1)] \\
    \geq & \frac{1}{c(j)}(c(l) - \lambda lf(l)c(l)) = \frac{f(l)c(l)}{c(j)}(\frac{1}{f(l)} - \lambda l) \\
    \geq & \frac{f(l)c(l)}{c(j)}(\lambda^* - \lambda) l\text{,}
  \end{IEEEeqnarray}
  where the last inequality holds due to the requirement that \(f(l) \leq \frac{1}{l}f(1)c(1)\max_{k\in N} \frac{k}{c(k)} = \frac{1}{l\lambda^*}\). If \(l > j > 0\), then \(l \geq j+1\) and \(f(l)c(l) \geq f(j+1)c(j+1)\), and the first inequality follows in the reasoning,
  \begin{IEEEeqnarray*}{Rl}
    &\frac{c(l)}{c(j)} - \frac{\lambda}{c(j)}[lf(j+1)c(j+1)] \\
    \geq & \frac{f(j+1)c(j+1)}{c(j)}\left(\frac{1}{f(l)}-\lambda l\right) \geq \frac{f(l)c(l)}{c(j)}(\lambda^* - \lambda) l\text{,}
  \end{IEEEeqnarray*}
  where the second inequality follows due to the same reasoning as above. As \(\lambda^* \geq \lambda\), and since \(f\) and \(c\) are nonnegative, it follows that \(\frac{f(l)c(l)}{c(j)}(\lambda^* - \lambda) l \geq 0\) and \(\frac{f(l)c(l)}{c(j)}(\lambda^* - \lambda) l \geq 0\).
  
  Case \(j+l > n\): Similar to the above, we must show that,
  \[ \lambda jf(j) \leq \frac{c(l)}{c(j)} + \frac{\lambda}{c(j)}[(n-l)f(j)c(j)-(n-j)f(j+1)c(j+1)], \]
  where the slope on the right-hand side is negative. This can be shown once again using the requirement that \(f(l) \leq \frac{1}{\lambda^* l}\), \(j+l-n > 0\), and the requirement that \(f(j)c(j)\) is nondecreasing for all \(j \in N\).
\end{proof}

Note that if the requirements for Corollary \ref{cor:nondecr} ii) are met, then the value of $\mu^*$ is determined by the strictest constraint at $\lambda = \lambda^*$. The expression for the price of anarchy can therefore be written explicitly as $\rm{PoA} = 1/C^*$, where $C^*$ is the minimum over the constraints, i.e.
\begin{align*}
    C^*\!{=} \min \begin{cases}
            \displaystyle \min_{1\leq j+l \leq n} \frac{c(l)}{c(j)} \:{+} \lambda^* \Big[ jf(j) \!{-} lf(j\!{+}1)\frac{c(j\!{+}1)}{c(j)}\Big]\\
            \displaystyle \min_{j+l>n} \frac{c(l)}{c(j)} \:{+} \lambda^* \Big[ (n \!{-} l)f(j)c(j) - \dots\\
            \qquad \qquad \qquad \qquad \dots\!{-} (n \!{-} j)f(j\!{+}1)\frac{c(j\!{+}1)}{c(j)}\Big]\text{,}
        \end{cases}
\end{align*}
where $j \neq 0$ and $j,l \in [0,n]$.

%% file: optimizing.tex
\section{Optimizing the Price of Anarchy}

We have developed a linear program for calculating the exact price of anarchy of a game given the number of players $n$, the cost function $c$ and the distribution rule $f$. In this section, we develop a framework that selects the optimal cost distribution rule for minimizing the price of anarchy. We show that this framework can also be posed as a linear program.

\begin{theorem}\label{thm:opt_dist_rule}
    Given the cost function $c$ and set $\mathcal{F}$ of permissible distribution rules, the optimal distribution rule
    \[ f^* = \argmin_{f \in \mathcal{F}} \PoA(f) \text{,}\]
    is in fact the solution to the following linear program 
    \begin{align}
        &f^* \in \argmax_{f \in \mathcal{F}, \mu \in \mathbb{R}} \mu\nonumber\\
        \text{s.t. } & 1_{\{b+x\geq1\}}c(b\!{+}x) \!{-} \mu 1_{\{a+x\geq1\}} c(a\!{+}x) \!{+} af(a\!{+}x)c(a\!{+}x) \nonumber\\
        & \quad\!{-} bf(a\!{+}x\!{+}1)c(a\!{+}x\!{+}1) \geq 0, \forall (a,x,b) \in \mathcal{I_R}\text{,} \label{eq:opt_distr_rule}
    \end{align}
    and the optimal price of anarchy is
    \[ \PoA(f^*) = \frac{1}{\mu^*}, \]
    where $\mu^*$ is the corresponding result in \eqref{eq:opt_distr_rule}.
\end{theorem}
\begin{proof}
    For the proof that $\argmin_{f \in \mathcal{F}} \PoA(f)$ is attained, meaning that the problem is well posed, see \cite[Lemma 4]{paccagnan2018distributed}. Solving this problem means finding the distribution rule that maximizes $C^*$ in Theorem \ref{thm:dual},
    \begin{align*}
        f^* \in \argmax_{f \in \mathcal{F}} &\max_{\lambda \in \mathbb{R}, \mu \in \mathbb{R}} \mu\\
        \text{s.t. } & 1_{\{b+x\geq1\}}c(b+x) - \mu 1_{\{a+x\geq1\}} c(a+x) \\
        & \quad+ \lambda e(a,x,b) \geq 0, \forall (a,x,b) \in \mathcal{I_R}\text{,}
    \end{align*}
    where $e(a,x,b)$ is defined as in Theorem \ref{thm:dual}. To avoid having to solve a nonlinear program, we combine $\lambda$ and $f$ in $\Tilde{f}(j) := \lambda f(j)$, and note that when $(a,x,b) = (0,0,1)$, then $\Tilde{f}(1) = \lambda f(1) \leq 1$, and that $\lambda \geq 1/f(1) > 0$ since $f>0$. We can now merge the two $\max$ operators to get
    \begin{align*}
        &\Tilde{f}^* \in \argmax_{\Tilde{f} \in \mathbb{R}_{\geq0}^n, \Tilde{f}(1) \geq 1, \mu \in \mathbb{R}} \mu\\
        \text{s.t. } & 1_{\{b+x\geq1\}}c(b+x) - \mu 1_{\{a+x\geq1\}} c(a+x) \\
        & \quad+ \lambda \Tilde{e}(a,x,b) \geq 0, \forall (a,x,b) \in \mathcal{I_R}\text{.}
    \end{align*}
    where $\Tilde{e}(a,x,b) = a\Tilde{f}(a\!{+}x)c(a\!{+}x)- b\Tilde{f}(a\!{+}x\!{+}1)c(a\!{+}x\!{+}1)$ for all $(a,x,b) \in \mathcal{I}$. We note that $\Tilde{f}^*$ must be feasible as $\Tilde{f}^* = \lambda f^*(j)$ and $f^* \in \mathcal{F}$, and that $\PoA(\Tilde{f}^*) = \PoA(f^*)$ as equilibrium conditions are invariant to scaling. 
\end{proof}

We have successfully developed a linear program approach to designing distribution rules that minimize the price of anarchy of a cost-minimization game, given the maximum number of agents in the game and the cost function. Next, we motivate these results by examining a class of cost-minimization games.

%% file: specialcase.tex
\section{Special Case: Convex Cost Functions} \label{sec:convex}
In this section, we explore the implications of the linear program developed in Section \ref{sec:lin-prog} for a particular class of cost-minimization games. We consider the class of games for which $c(j)$ is convex and nondecreasing, i.e. for all $j \in [n-1]$,
\begin{align*}
    c(j+1) &\geq c(j) \\
    c(j+1)-c(j) &\geq c(j)-c(j-1)\text{.}
\end{align*}

Note that these properties of $c$ imply that $f_{\textrm{MC}}(j) \geq f_{\textrm{SV}}(j)$, as $c(j) - c(j-1) \geq \frac{c(j) - c(0)}{j} \text{, } \forall j > 1$, and we assumed earlier that $c(0) = 0$. Only in this section, we also assume, without loss of generality, that $f(1)c(1) = 1$.

\begin{theorem} \label{thm:c_convex}
    For $f(j)c(j)$ nondecreasing, $c(j)$ convex and nondecreasing, and $f(j) \leq f_{\textrm{MC}}(j) \text{, } \forall j \in N$, we can rewrite Corollary \ref{cor:nondecr} as $\PoA(f) = 1/C^*$, where
    \begin{IEEEeqnarray*}{rCl}
        C^* & = & \min_{\lambda \in [0, 1], \mu \in \mathbb{R}} \mu \\
        & \text{s.t.} & \mu c(j) \!\leq\! c(l)\! +\! \lambda [\min\{j, n\!-\!l\}f(j)c(j)\! \\
        && - \!\min\{l, n\!-\!j\}f(j\!+\!1)c(j\!+\!1)]\text{,} \quad l \leq j \in N
    \end{IEEEeqnarray*}
\end{theorem}
\begin{proof}
    Observe that if we were to write this theorem's claim for $j+l\leq n$ and $j+l>n$, separately, it would resemble Corollary \ref{cor:nondecr}(i), with $\lambda = 1$ and $j \geq l$. Note that when \(j=0\) and \(1 \leq j+l \leq n\), the constraints in Corollary \ref{cor:nondecr} become,
    \[ \lambda \leq \frac{1}{f(1)\,c(1)}\,\frac{c(l)}{l}. \]
    Due to the concavity of \(c\), the right-hand side is minimized for \(l = 1\), and $\lambda \leq 1/f(1)c(1) = 1$, by assumption. We continue by demonstrating that the constraints on $\mu$ when $j = l$ are more strict than those when $j < l$, for $\lambda \leq 1$.
    
    When $1 \leq j+l \leq n$, it is sufficient to show that,
    \[ 0 \leq c(l) - c(j) + \lambda (j-l)f(j+1)c(j+1)\text{.}\]
    For convex, nondecreasing cost function $c$, $l > j$, and $\lambda \leq 1$,
    \[ c(l) \geq c(j) + \lambda(c(j+1) - c(j))(l-j)\text{.}\]
    Thus,
    \begin{align*}
        &c(l) - c(j) + \lambda (j-l)f(j+1)c(j+1)  \\
        &\geq c(j) + \lambda(c(j+1) - c(j))(l-j) - c(j)\\
        & \quad - \lambda (l-j)f(j+1)c(j+1) \\
        &= \lambda(l-j)(c(j+1) - c(j) - f(j+1)c(j+1)) \geq 0\text{,}
    \end{align*}
    where the final inequality relies on $c(j+1) - c(j) \geq f(j+1)c(j+1)$, which is true for $f(j) \leq f_{\textrm{MC}}(j) \text{, } \forall j \in N$, and $l>j$. Note that we need not consider the case when $j = n$ because there is no $l \in N$ such that $l > j = n$.
    
    The proof that $\mu$ is more strictly constrained for $l=j$ than for $l>j$, when $j+l > n$ is almost identical, and also uses the requirement that $f(j) \leq f_{\textrm{MC}}(j) \text{, } \forall j \in N$.
\end{proof}

The expression for $C^*$ in Theorem \ref{thm:c_convex} may further simplify when the distribution rule is known. For the two distribution rules we have formally defined in this paper, we can explicitly write expressions for the price of anarchy. We have that $\PoA(f_{\textrm{SV}}) = 1/C_{\textrm{SV}}^*$, where,
\begin{align*}
C_{\textrm{SV}}^* & {=}  \min_{l \leq j \in N} \left\{ \frac{c(l)}{c(j)} {+} \frac{\min\{j,n{-}l\}}{j}{-}\frac{\min\{l, n{-}j\}c(j {+}1)}{(j{+}1)c(j)} \right\}
\intertext{and $\PoA(f_{\textrm{MC}}) = 1/C_{\textrm{MC}}^*$, where,}
C_{\textrm{MC}}^* & {=} 1 {+} \min_{j \in N} \left\{ \frac{\min\{j,n{-}j\}}{c(j)} [2c(j){-}c(j{-}1){-} c(j{+}1)] \right\}. 
\end{align*}
The expression for $C_{\textrm{SV}}^*$ results from substituting $f(j) = 1/j$ into the equation in Theorem \ref{thm:c_convex}, and from Corollary \ref{cor:nondecr} ii). The proof for $C_{\textrm{MC}}^*$ is left out for conciseness, but amounts to showing that the constraints for $j > l$ are stricter than those for $j=l$, due to similar reasoning as showing that we need only consider $j \geq l$ in the proof of Theorem \ref{thm:c_convex}.

%% file: appendix.tex
\section{Appendix}

\begin{lemma}\label{lem:const_geq}
    The optimal value of the relaxed problem, $D^*$, is greater than or equal to $C^* = 1/\PoA(f)$.
\end{lemma}
\begin{proof}
We prove the lemma by constructing a game in the parametrization that satisfies the constraints of the original problem. Setup a congestion game with a feasible point $\theta(a,x,b)$ with value $v$. This game consists of the resources $r(a,x,b,i)$ for all $i \in N$ and for all $(a,x,b)\in \mathcal{I}$. Each resource $r(a,x,b,i)$ is assigned the value $\theta(a,x,b)/n$. In the Nash equilibrium strategy, each agent $i$ selects $a+x$ consecutive resources from each set $\{r(a,x,b,j)\} \forall j \in N, (a,x,b) \in \mathcal{I}$ starting with resource $r(a,x,b,i)$. In the optimal strategy, each agent $i$ selects $b+x$ consecutive resources from the sets $\{r(a,x,b,j)\} \forall j \in N, (a,x,b) \in \mathcal{I}$, starting from resource $r(a,x,b,(i-b) \mod n)$. It is important to note that the allocations $a_i^{\nel}$ and $a_i^{\opt}$ have $|a_i^{\nel} \cap a_i^{\opt}| = x$.

We must verify that $a^{\nel}$ is a Nash equilibrium. This game is known to be a potential game, with potential function,
\[ \phi(a) = \sum_{r\in \mathcal{R}}\sum^{|a|_r}_{j=1}v_rf(j)c(j). \]
Now we show that $\phi(a_i^{\opt}, a_{-i}^{\nel}) - \phi(a^{\nel}) \geq 0 \quad \forall i \in N$, which implies that $J_i(a_i^{\opt}, a_{-i}^{\nel}) - J_i(a^{\nel}) \geq 0 \quad \forall i \in N$ by the definition of potential games. We have that
\begin{align*}
    \phi(a^{\nel}) &= \sum_{j \in N} \sum_{a,x,b}\frac{\theta(a,x,b)}{n}\sum_{i=1}^{a+x}f(i)c(i) \\
    &= \frac{1}{n} \sum_{a,x,b} n \theta(a,x,b)\sum_{i=1}^{a+x}f(i)c(i)\text{.}
\end{align*}
When the agent $i$ deviates from allocation $a^{\nel}$ to $(a_i^{\opt}, a_{-i}^{\nel})$, there is one additional agent selecting $b$ resources, one less agent selecting $a$ resources, and $n-a-b$ resources that are selected by the same number of agents. We can therefore write that,
\begin{align*}
    &\phi(a_i^{\opt}, a_{-i}^{\nel}) - \phi(a^{\nel}) = \\
    {=}\;& \frac{1}{n} \sum_{a,x,b} \theta(a,x,b) \Big(b \sum_{i=1}^{a+x+1}f(i)c(i)+ a \sum_{i=1}^{a+x-1}f(i)c(i)\\
    &{+}\: (n-a-b) \sum_{i=1}^{a+x}f(i)c(i) \Big) - \phi(a^{\nel})\\
    {=}\;& \frac{1}{n} \sum_{a,x,b} \theta(a,x,b) \Big(bf(a+x+1)c(a+x+1)\\
    &{-}\: af(a+x)c(a+x) \Big).
\end{align*}
Which must be greater than or equal to zero, by both the original and the relaxed constraints.
\end{proof}